\documentclass[journal,letterpaper, onecolumn ]{IEEEtran}
\pdfoutput=1
\usepackage[]{amsmath}
\usepackage{amssymb}
\usepackage{fullpage}
\usepackage{amsthm}
\usepackage{bbold}
\usepackage{graphicx}
\usepackage{algpseudocode}
\usepackage{algorithm}
\usepackage{pdfsync}
\usepackage{url}
\usepackage{flushend}
\usepackage{xcolor}
\usepackage{bm}
\usepackage{tikz}
\usepackage{caption}
\usepackage{float}
\usepackage{subcaption}

\newcommand{\txtemail}[1]{\textbf{\texttt{#1}}}
\newcommand{\mb}[1]{\mathbf{#1}}
\newcommand{\bma}[1]{\boldsymbol{#1}}

\newcommand{\op}[1]{\operatorname{#1}}
\theoremstyle{newstyle}
\newtheorem{theorem}{Theorem}
\newtheorem{lemma}{Lemma}
\newtheorem{definition}{Definition}

\begin{document}

\sloppy

\title{Low-complexity Decoding is Asymptotically Optimal in the SIMO MAC }

\author{
  \IEEEauthorblockN{Mainak Chowdhury and Andrea Goldsmith}\\

 \thanks{ The authors are with the Department of Electrical Engineering,
	 Stanford University, Stanford, CA - 94305. Questions or comments can
	 be addressed to \txtemail{\{mainakch,andreag\}@stanford.edu}.
	 Parts of this work were presented at ISIT, 2013 and Allerton, 2013. This work is supported by the 3Com Corporation Stanford
	 Graduate Fellowship, the NSF Center for
Science of Information (CSoI): NSF-CCF-0939370 and by a gift from Cablelabs. }}

\maketitle


\begin{abstract}

  A single input multiple output (SIMO) multiple access channel, with
  a large number of transmitters sending symbols from a constellation
  to the receiver of a multi-antenna base station, is considered. The
  fundamental limits of joint decoding of the signals from all the
  users using a low complexity convex relaxation of the maximum
  likelihood decoder (ML, constellation search) is investigated.  It has
  been shown that in a rich scattering environment, and in the
  asymptotic limit of a large number of transmitters, reliable
  communication is possible even without employing coding at the
  transmitters.  This holds even when the number of receiver antennas
  per transmitter is arbitrarily small, with scaling behaviour
  arbitrarily close to what is achievable with coding. Thus, the
  diversity of a large system not only makes the scaling law for coded
  systems similar to that of uncoded systems, but, as we show, also allows
  efficient decoders to realize close to the optimal performance of maximum-likelihood decoding. However, while there is no performance loss relative to the scaling laws of the optimal decoder, our 
  proposed low-complexity decoder exhibits a loss of the
  exponential or near-exponential rates of decay of error probability
  relative to the optimal ML decoder.
\end{abstract}

\section{Introduction}

Although the capacity-achieving techniques of superposition coding at the
encoder and joint decoding at the decoder \cite{cover1991elements} promise
significantly higher capacity for multiuser networks, such sophisticated coding
schemes often suffer from practical challenges. Thus the simpler 
orthogonal schemes to separate users either in time (TDMA), space (sectorization in cellular networks) or frequency (FDMA) have remained in widespread use.  Moreover, the capacity benefits of the optimal scheme over orthogonalizing schemes like time-division have been shown to be
negligible in some regimes such as under asymptotically  low-power or for asymptotically many users (\cite{caire2004suboptimality},\cite{tse2005fundamentals}).

In this work we consider a multiple access setting where an
asymptotically large number of transmitting users communicate in a
rich scattering environment with a single multi-antenna base
station. We look at transmitting schemes which do not employ coding,
but instead transmit symbols from the BPSK constellation and rely on 
the diversity inherent in a large system to achieve reliability. Such
a setting may model sensor networks or general distributed networks
with energy/processing power limitations at the transmitters (which
may preclude sophisticated coding schemes) and centralized
receivers. A similar setting was considered in the companion paper
\cite{ chowdhury2014uncodedfundamental}, where it was shown that using the optimal
maximum likelihood (ML) decoder, the decoding can be made arbitrarily
reliable for an arbitrarily small number of receiver antennas per
transmitter, provided that the number of transmitters is large enough.
A similar setup was also considered in \cite{chowdhury2013reliable}
where a relaxation of the maximum likelihood decoder was shown to be
asymptotically reliable (i.e. the probability of error vanishing to
zero) in the number of transmitters, provided that the number of
 receiver antennas is more than the number  of transmitters.

In this work, we analyze the same low-complexity decoder proposed in
\cite{chowdhury2013reliable} and modify it to handle underdetermined
systems, i.e. systems where the number of receiver antennas is less than the
number of transmitter antennas.  In particular we consider the decoder obtained
by expanding the search over possible symbols to intervals instead of discrete
points and then quantizing the output of the interval search to the nearest
constellation point. This relaxation of the search over integer points to
search over intervals allows more efficient (polynomial time) decoding, but
may not be unique in the regime of underdetermined systems (because of the non-trivial null space for a wide channel matrix). Thus the above procedure yields
a non-singleton solution set in general. We propose a family of randomization
techniques and show that they can  return  provably good estimates from this
solution set. Henceforth we will refer to this decoding technique as the
randomized interval search and quantize (r-ISQ) decoder.  We obtain analytical
bounds on the performance of the r-ISQ decoder and  show that reliable decoding
(in a sense made precise in the later sections) is possible in the asymptotic
limit of a large number of transmitters and receivers, with the per-transmitter
number of receiver antennas being held constant at \emph{any} arbitrary positive
value.  Using the same techniques used in the proof however, the
per-transmitter number of receive antennas can be shown to be arbitrarily close
to the theoretically optimal scaling derived e.g. in
\cite{mainakch2012uncoded},\cite{chowdhury2014uncodedfundamental}.

The rest of the paper is organized as follows. We first present the system
model and describe the optimal decoder and the r-ISQ decoder.  We then describe
a bound on the error probability of this decoder. Asymptotic analyses of these
bounds are then presented.

\section{System Model}

\begin{figure}[h] 

\begin{center}
		\includegraphics[width=\linewidth]{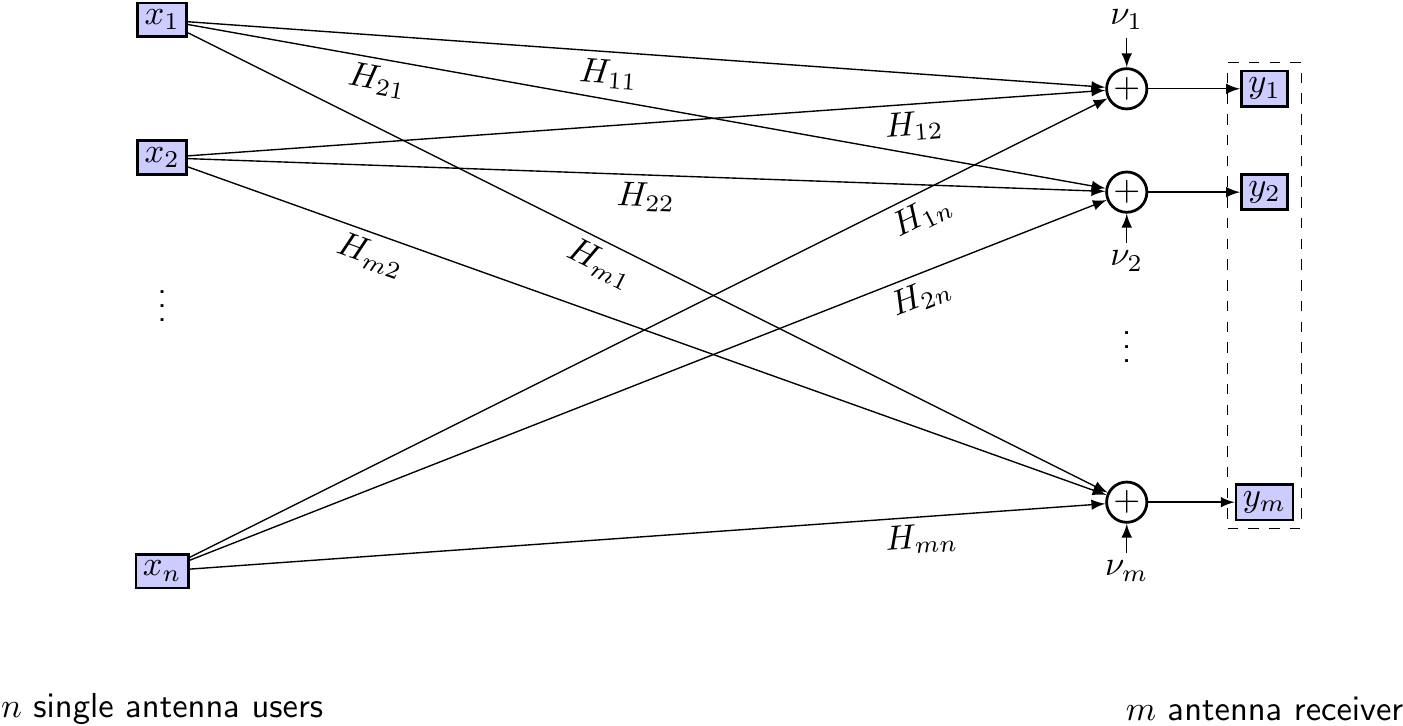}
	\end{center}
	\caption{System model}
	\label{fig:1}
\end{figure}

Our system model is depicted in Figure \ref{fig:1}. 
We have an uplink system with $n$ single-antenna transmitters and an $m$ antenna  receiver. The channel matrix $\mb{H} \in \mathbb{R}^{m \times n}$ is chosen to
model a rich scattering environment, and the entries are assumed to be drawn
i.i.d. from $\mathcal{N} (0,1)$. The $k^{th}$ column of $\mb{H}$ is denoted as $\mb{h_k} \in \mathbb{R}^m$. Thus \[\mb{H}= \begin{pmatrix} \mb{h_1} & \mb{h_2} & \ldots & \mb{h_{n-1}} & \mb{h_n}
\end{pmatrix}.\] We also assume that  the users do not cooperate with each
other and that they transmit symbols from the standard unit energy BPSK
constellation. The components of the noise at the receiver (\boldmath $\nu$ \unboldmath) are assumed
to be i.i.d. $\mathcal{N}(0,\sigma^2)$. The received signal at the multi-antenna receiver is then 
\begin{IEEEeqnarray}{l+}
	\mb{y} = \mb{H x} +  \bma{\nu} .
	\label{eq:1}
\end{IEEEeqnarray}
The vector $\mb{x}\in \{-1,+1\}^n$, which consists of the transmitted symbols from
the $n$ users, is referred to as the \emph{n-user codeword} to indicate that the
receiver decodes the block of n-user constellation points simultaneously. We further assume
that the receiver has perfect channel state information (CSI) and that the
transmitters have no CSI.

\section{Previous work and results }

We now describe a few observations and results about the performance limits of
this system. These results, derived in \cite{mainakch2012uncoded}, \cite{chowdhury2014uncodedfundamental}, assume that
the receiver employs ML decoding, i.e. it returns 

\begin{IEEEeqnarray}{l+}
	\mb{\hat{x}} =\operatorname{argmin}_{\mb{x} \in \{-1,+1\}^n} ||\mb{y}-\mb{H x}|| ^2.
	\label{eq:2}
\end{IEEEeqnarray}
With this decoder it has been shown that in the limit of a large number of transmitters the following holds:
\begin{theorem}
	Under ML decoding, there exists a $d>0$ such that  for all sufficiently large $n$,
	the probability of error in decoding the \emph{n-user codeword} satisfies
	\[{P_{\mathrm{error}}} \leq 2^{-d n}.\]
	\label{thm:log2}
\vspace{-.5cm}
\end{theorem}

In other words the probability of decoding a particular user's transmitted
symbol in error decreases exponentially with the number of users, even though
the users do not employ any coding across time.

A critical component in \cite{mainakch2012uncoded}, \cite{chowdhury2014uncodedfundamental} to achieve this asymptotic
result is the use of ML decoding. In this work we investigate whether we can
achieve reliability even with lower-complexity decoders. In particular, we ask
whether an efficient polynomial time decoder can realize an asymptotically
vanishing probability of error, as was the case with the ML decoder.  

A common approach to relax hard combinatorial optimization problems (such as ML
decoding) is the technique of expanding the search space from discrete points
to intervals or regions \cite{boyd2004convex}. Motivated by this idea, we
consider a convex relaxation of the maximum likelihood decoder search as
follows: 

\begin{IEEEeqnarray}{l+} \mb{\hat{x}}
	=\operatorname{sgn}(\operatorname{argmin}_{\mb{x} \in [-1,+1]^n}
	||\mb{y}-\mb{H x}|| ^2).
	\label{eq:3}
\end{IEEEeqnarray}

In the above $\operatorname{sgn}(\mb{x})$ for $\mb{x} \in \mathbb{R}^n$ refers  to the vector obtained by the coordinatewise application of the signum function defined below for a scalar $x$. \[\operatorname{sgn}(x) = \begin{cases}
	1 & \text{ if $x>0$}.\\
	-1 & \text{ otherwise}.
      \end{cases}\] 

      The modified decoder in (\ref{eq:3}) expands the search for a
      valid \emph{n-user codeword} to the interval $[-1,1]$ per
      dimension and then quantizes it to integer values afterwards,
      hence we call it an ISQ decoder.  This idea of relaxing an
      integer program to a box-constrained program is a well known
      technique and has been studied in different settings, e.g. in
      \cite{donoho2010counting}, \cite{bayati2011dynamics}, where
      different asymptotic properties of this decoder are
      established. While some of the results (especially the
      characterization of the null space of Gaussian random matrices
      \cite{donoho2010counting}, and the behaviour of approximate
      message passing (AMP) type algorithms \cite{bayati2011dynamics}
      with the box constraints) from these works do give us insights
      into the expected behaviour of the box constrained decoder in
      some regimes (e.g. $\frac{m}{n} > 0.5$ with BPSK transmissions),
      the regime of an arbitrarily small fraction of the
      per-transmitter number of receiver antennas is still not fully
      characterized. In fact, for the AMP decoder, it can be shown
      that if $\frac{m}{n} < 0.5$, the number of symbol errors in the
      decoded block would be $\Theta (n)$, i.e. the number of
      incorrectly decoded symbols is linear in the number of
      transmitting users. We consider a slight modification of the
      box-constrained decoder and with this modification, are able to
      show asymptotic reliability in a sense made precise below.

      Note that in the regime where the channel matrix is
      underdetermined (i.e. $m < n$ ) the above procedure may not give
      a unique solution. If an \emph{n-user} codeword $\hat{\mb{x}}$
      is a solution, then any codeword of the form $\hat{\mb{x}} +
      \bma{Z \beta}$ is also a solution. Here $\mb{Z}$ is a basis for
      the right null space of $\mb{H}$, i.e. $\mb{Z}$ is such that
      $\mathbf{H}\mb{ Z} = \mathbf{0}$ and $\bma{\beta} \in
      \mathbb{R}^{(1-\alpha) n}$. Thus the ISQ decoder, in this case,
      cannot uniquely specify a solution by itself. It would, in
      general, give an affine subspace as a solution. In order to
      specify a unique solution, we propose a randomization step
      (randomized ISQ or r-ISQ). Specifically, we propose a family of
      distributions and show that, for estimates drawn according to
      this general family of distributions, we can achieve reliability
      in a sense which is made precise in the following sections. We
      further show that it is possible to sample efficiently from a
      member of this family of distributions. For specificity let us
      consider the following decoder.
\begin{IEEEeqnarray}{l+} \text{r-ISQ:} \,\mb{\hat{x}}
	=\operatorname{sgn}(\operatorname{argmin}_{\mb{x} \in S }
	||\mb{x_r}-\mb{ x}||_{\infty}). \nonumber
	\label{eq:3a}
\end{IEEEeqnarray}
In the above $S = \{\mb{x}: \mb{x} \in \op{argmin}_{\mb{z}}||\mb{y} - \mb{Hz}||\}$, and $\mb{x_r} \sim \op{Unif}([-1,1]^n)$. Since both $\mb{x_r}$ and $S$ can be computed efficiently (in polynomial time, for reasons discussed in later sections), we see that there is an efficient algorithm to achieve asymptotic reliability without employing coding.

However there is a performance hit, relative to ML decoders, when we move to
r-ISQ decoders. This is in terms of the rate of decay of the error probability
with the number of transmitting users. Although the error probability seen by
each user vanishes to zero, we do not have an upper bound for the probability of
having at least one symbol error in the \emph{n-user} codeword, which is in
contrast to ML, where the block error probability decays exponentially.  This
lack of exponential decay with the simpler decoder is primarily due to the self
interference due to the search over intervals. Thus, in particular, we do seem
to lose the exponential fall off in the probability of error that is achieved with the ML decoder.  However, the number of symbol errors in the
decoded block in the  asymptotic limit is at most sub linear, i.e. the
probability that a constant fraction of the transmitter symbols are incorrectly
decoded can be made arbitrarily small for a large enough system size. Defining
$P_e^{k'}$ to be the probability of incorrectly decoding  at least $k'n$ out of
$n$ transmitted symbols, the following states a bound on $P_e^{k'}$.

\begin{theorem}
	Under r-ISQ decoding, for $m = \alpha n$ , $\alpha > 0$ and any constant $k>0$, there exists a $d>0$ such that
	for all sufficiently large  $n$,
	\[P_{e}^{k} \leq 2^{-d n \log n}.\]
	\label{thm:main}
\vspace{-0.5cm}
\end{theorem}

We mention that the same proof techniques used to establish the
above result can also be used to get sharper bounds, i.e. for $k =
\frac{1}{n^{\gamma}}$ for some $0<\gamma <1$. Thus the per-user error
probability is asymptotically less than $\frac{1}{n^{\gamma}}$. Also, while
in this section we focus on the case where the ratio ($\alpha$) of the number of receiver
antennas to the number of transmitter antennas is constant, we point out that
the same techniques continue to hold even for  \[\alpha_n =
\frac{1}{(\log n)^{\xi}}, \xi<1.\] By making $\xi$ close to $1$ we see that we
can come arbitrarily close to the optimal scaling established for the ML
decoder in \cite{mainakch2012uncoded}.  Thus we see that by exploiting the
\emph{diversity} (richness in the scattering environment), one can not only
get arbitrarily reliable communication in different asymptotic regimes (in this
case for a large number of transmitters) without employing coding or ML
decoders, but can also achieve optimal scaling for the per transmitter  number of receiver
antennas.

We now use the bound in Theorem \ref{thm:main} to bound the probability of symbol error seen by each transmitter.

\begin{theorem} The probability of error with the r-ISQ decoder seen by any
transmitting user vanishes in the limit of an asymptotically large number of
transmitters, with the per-transmitter number of receive antennas being any
constant $\alpha > 0$.
	\label{thm:main2}
\end{theorem}


The remainder of this paper discusses the proofs of Theorems \ref{thm:main} and \ref{thm:main2} and points out suitable generalizations using the same proof techniques.

\section{An Upper Bound on the Decoding Error}

We first present an upper bound on the probability of decoding error. Most of
the steps described in this are similar to what was used in
\cite{chowdhury2013reliable}, modified to take into account the fact that there
is a non-trivial null space (so the solution set in general may not be unique).
We look at the pairwise error probability of mistaking the  transmitted
codeword with one differing in $k' n$ symbols. For (\ref{eq:3}), the
probability of mistaking a codeword $\mb{x_0}$ for another differing in $i$
symbol positions is given by

\begin{IEEEeqnarray}{rCl}
	P_{e,\mb{b_i}} &\leq& Q \left(\min_{{\mb{x}}:\operatorname{supp}(\operatorname{sgn}(\mb{x})-\mb{x_0})=\mb{b_i}} \frac{|| \mb{H}(\mb{x}-\mb{x_0})||}{2\sigma}\right) .
\end{IEEEeqnarray}

Here $Q(x) = \frac{1}{\sqrt{2\pi}}\int_x^\infty e^{-x^2/2} dx$, $\mb{b_i} $ is
a vector of size $i$ whose entries are positions where the codewords differ
(arranged in increasing order), and $\mb{b_i}(j)$ is the $j^{th}$ symbol
position where the codewords differ. We point out here that $\mb{b_i}$ has a
one-to-one correspondence with a subset of $\{1,\ldots,n\}$ of cardinality $i$.
$||\mb{x}||_0$ refers to the number of non-zero entries in $\mb{x}$.
$\operatorname{supp}(\mb{x})$ refers to the support (i.e. locations of the non-zero entries of vector $\mb{x}$).

Note that the error probability above is independent of which $\mb{x_0}$ is
chosen, when averaged over the distribution of $\mb{H}$. Hence, choosing $\mb{x_0} = -\mb{ 1}$, we note that the last expression
can be rewritten as follows.

\begin{IEEEeqnarray}{C}
	P_{e,\mb{b_i}}\leq	Q \left(\min_{\substack{1\leq c_j \leq 2 \, \forall j \in \mb{b_i}\\ 0 \leq c_j \leq 1 \, \forall j \in \mb{b_i}^c }} \frac{||\sum_{j=1}^n c_j \mb{h_{b_i(j)}}||}{2\sigma}\right) \nonumber \\
	\overset{(a)}\leq \frac{1}{2}\exp{\left(-\min_{\substack{1\leq c_j \leq 2 \, \forall j \in \mb{b_i}\\ 0 \leq c_j \leq 1 \, \forall j \in \mb{b_i}^c }} \frac{||\sum_{j=1}^n c_j \mb{h_{b_i(j)}}||^2}{8 \sigma^2}\right)}
	\label{eq:5}
\end{IEEEeqnarray}
where (a) follows because $Q(x) \leq \frac{1}{2} \exp(\frac{-x^2}{2})$.
We observe now that  (\ref{eq:5}) averaged over the channel realizations is independent of the particular subset of symbols that are decoded in error and depends only on the size $i$ of such a subset. Let's call this averaged probability of error  $P_i$. Thus we have 
\begin{IEEEeqnarray*}{C}
	P_{i} \triangleq E_{\mb{H}} e^{\left(-\min_{\substack{1\leq c_j \leq 2 \, \forall \, j \in \{1,\ldots,i\}\\ 0 \leq c_j \leq 1 \, \forall j \in  \{i+1,\ldots,n\} }} \frac{||\sum_{j=1}^n c_j \mb{h_{j}}||^2}{8\sigma^2}\right)}.
\end{IEEEeqnarray*}

If $P_e^{k'}$ is the probability of error of decoding at least $k'n$ transmitter symbols incorrectly, and $S_i$ refers to the set of all vectors representing subsets of size $i$ from $\{1,\ldots,n\}$, then a union bound for the error probability is

\begin{IEEEeqnarray}{rCl}
	P_e^{k'} &\leq& \sum_{k'n \leq i \leq n} \sum_{\mb{b} \in S_i} \frac{1}{2} P_i\\
	&\leq& \sum_{k'n  \leq i \leq n } \binom{n}{i} \frac{1}{2} P_i.
	\label{eq:7}
\end{IEEEeqnarray}
Note that, by the symmetry of the system, the probability of error $P_e$ seen by each transmitting user is upper bounded by 
\[P_e \leq k' + P_e^{k'}.\]
We show that for any small $k'$, there exists a large enough system size for which $P_e^{k'}$ becomes exponentially small, even with a convex decoder of much lower complexity. This will establish Theorem \ref{thm:main2}.
\section{Asymptotic analysis of the upper bound}
We first prove bounds on the exponent appearing in the bound for $P_{e,\mb{b_i}}$ in (\ref{eq:5}). Specifically, we look at (ignoring a constant scaling of $8 \sigma^2$) 
\[\min_{\substack{1\leq c_j \leq 2 \, \forall j \in \{1,\ldots,i\}\\ 0 \leq c_j \leq 1 \, \forall j \in  \{i+1,\ldots,n\} }}  ||\sum_{j=1}^n c_j \mb{h_{j}}||^2.\]

Before we describe the proof, we define the ($\epsilon,\bma{\delta}$)-grid inside
the hypercube $[-1,+1]^n$, for some $0<\epsilon <0.25$. This grid is simply the
set of points \[\mathcal{G}_{n,\epsilon,\bma{\delta}}=\{\mb{x}:x_i \bmod \epsilon
= \delta_i, -1 \leq x_i\leq 1 \,\,\forall i\}.\] As an illustration, for $\bma{\delta} =
\bma{0}$, it may be rewritten as \[\mathcal{G}_{n,\epsilon,
\bma{\delta}}=\{-1,-1+\epsilon,\ldots,1-\epsilon,1\}^n.\] if
$\frac{1}{\epsilon} \in \mathbb{N}$. 

We now introduce the $(\epsilon,\bma{\delta})$-ISQ decoder, so named  because it replaces the interval search in the ISQ decoder by an ($\epsilon,\bma{\delta}$)-grid search:

\begin{IEEEeqnarray}{l+} \text{$(\epsilon,\bma{\delta})$-ISQ:} \,\mb{\hat{x}}_{\epsilon,\bma{\delta}}
	=\operatorname{sgn}(\operatorname{argmin}_{\mb{x} \in \mathcal{G}_{n,\epsilon,\bma{\delta}} }
	||\mb{y}-\mb{H x}|| ^2). \nonumber
	\label{eq:3a}
\end{IEEEeqnarray}

The $(\epsilon,\bma{\delta})$-grid error probabilities are defined similar to the
definitions for the ISQ decoder in the previous section, and are indicated by
an $\epsilon,\bma{\delta}$ subscript. We now collect some observations about
the grid error probabilities and use a union bounding argument for
$P_{e,\mb{b_i}, \epsilon,\bma{\delta}}$. Some of these results for the grid error
probabilities have already been derived in \cite{chowdhury2013reliable}  and have
been reproduced here for completeness and continuity of presentation. We
require the following lemma about the negative of the exponent in the error
probability $P_{e,\mb{b_i},\epsilon,\bma{\delta}}$:
\[\min_{\substack{ c_j \in [1,2], c_j \bmod \epsilon = \delta_j \, \forall j
\in \{1,\ldots,i\}\\ c_j \in [0,1], c_j \bmod \epsilon = \delta_j \, \forall j
\in  \{i+1,\ldots,n\} }}  ||\sum_{j=1}^n c_j \mb{h_{j}}||^2.\]

\begin{lemma} For any $i> k'n$, there exists an $n_0$ and an $a>0$, such that for
	all $n>n_0$, \[P(|| \sum_{j=1}^n c_j \mb{h_{j}}||^2 < a n \log n) \leq
		\exp(- a n \log n).\] 
                \label{lemma1}
\end{lemma}

\begin{proof}

We can show this using Markov's inequality. Let $a_1 = \frac{\alpha}{4}.$ Then 

\begin{IEEEeqnarray}{rCl} 
&&P(|| \sum_{j=1}^n c_j \mb{h_{j}}||^2 < a_1 n \log n)\\ 
&=& P(\exp(-t || \sum_{j=1}^n c_j \mb{h_{j}}||^2) > \exp(-t a_1 n \log n)) \nonumber \\ 
&& \overset{(a1)} \leq \exp(t a_1 n \log n) E \exp(-t || \sum_{j=1}^n c_j \mb{h_{j}}||^2) \label{eq:general} \\ 
&& \overset{(a2)}= \exp(t a_1 n \log n) (1+2t \sum_j c_j^2 )^{(-\alpha n/2)}  \\ &&\overset{(b)} \leq  \exp(t a_1 n \log n) (1+2t k' n )^{(-\alpha n/2)}   \\ 
&& \overset{(c)} \leq \exp(-\tilde{a} n \log n) \,\, \text{ for large enough n} \label{eq:gaussian}.  
\end{IEEEeqnarray} 

In the above $(a1)$ follows from Markov's
inequality, $(a2)$ follows from the moment generating function of a chi-squared
random variable, $(b)$ follows from the fact that for at least $k' n$ errors,
\[\sum_j c_j^2 \geq k'n,\] and $(c)$ follows by choosing $t=1$, and defining e.g.
$\tilde{a} = \frac{\alpha}{4}.$ Defining $a = \min(a_1, \tilde{a}) =
\frac{\alpha}{4}$, we get the claim in the lemma. Thus the claim is established
for $\mb{H}$ with $\mathcal{N}(0,1)$ entries. 


\end{proof} 
By observing that for a positive r.v., $P(x<d_0)< \exp(-d_0) $
implies 
\begin{IEEEeqnarray}{rCl}
&&E(\exp(-x)) \\&\leq& \exp(-d_0) + (1-\exp(-d_0))\exp(-d_0) \\ &\leq& 2 \exp(-d_0),	
\end{IEEEeqnarray}

we get that 

\begin{IEEEeqnarray}{rCl}
P_{i,\epsilon, \bma{\delta}} &\leq& E (\exp( -|| \sum_{j=1}^n c_j \mb{h_{j}}||^2)) \\
&\leq& \exp(-a n \log n) \,\, \text{ for large enough n}.
\end{IEEEeqnarray}

The probability of the event that there are at least $k' n$ symbols that are
decoded incorrectly can then be union bounded as follows.

\begin{IEEEeqnarray}{rCl}
	P_{e,\epsilon, \bma{\delta}}^{k'} &\leq& \sum_{i=k'n}^n \binom{n}{i} \left(\frac{1}{\epsilon} \right)^n P_{i,\epsilon,\bma{\delta}} \\
	& \overset{(d1)} \leq & n 2^{n(\max_{k'n \leq i \leq n} H_2(\frac{i}{n}) -\log(\epsilon) - a  \log n )} \,\, \text{ for }\nonumber  \\ && \text{ $a>0$ and large enough $n$} \\
	& \overset{(d2)} \leq & 2^{-a n \log n} \,\,\text{for a large enough n}. \nonumber
	\label{eq:11}
\end{IEEEeqnarray} 

In the above, \[H_2(x) = -x \log x - (1-x) \log(1-x).\]$ (d1)$ follows by
noting that \[\binom{n}{i} \leq 2^{H_2(i/n)}\] and $(d2)$ follows from the
observation that $H_2(\cdot)$ is bounded above by a constant. We now note that
by introducing an arbitrary distribution $f(\bma{\delta})$ on $\bma{\delta}$, i.e. randomizing the grid, there would be a
distribution induced on $\bma{\hat{x}}_{\epsilon,\bma{\delta}}$. Let's call that $\hat{f}(\bma{\hat{x}}_{\epsilon,\bma{\delta}})$. Thus
statements about the probability of error associated with  $\mathbf{\hat{x}}_{\epsilon,\bma{\delta}}$ would continue to
hold even for samples $\mathbf{y}$ drawn from $\hat{f}(\mathbf{y})$. Note that
sampling from this distribution may still be of exponential complexity. Let's
call this decoder the r-$(\epsilon,\bma{\delta})$ ISQ decoder, where r stands for
randomized.

We now relate the solution from the search over the randomized $(\epsilon,\bma{\delta})$-grid $\mathcal{G}_{n,\epsilon,\bma{\delta}}$
(i.e. the output of the r-$(\epsilon,\bma{\delta})$ ISQ decoder) to the solution ($\hat{\mb{x}}$) of the r-ISQ
decoder. Note that, in general, the ISQ decoder will not be unique and there is always an uncertainty due to the right null space of $\mathbf{H}$. Thus if the objective function in the ISQ decoder attains its infimum at $\hat{\mb{x}}$, then it will also attain the same infimum at all points of the following solution set \[S = \{ \mb{x} : \mb{x} = \hat{\mb{x}} + \mb{Z} \bma{\beta}, \mb{HZ} = \bma{0},\bma{\beta} \in \mathbb{R}^{(1-\alpha)n}\}.\] We show
next that a certain randomized choice of solution from this solution set will
be ``good''. Before that however, we introduce some notation. Let the projection of any vector $\mb{y}$ on any set $A$ be defined by
\begin{IEEEeqnarray}{rCl}
	P_A(\mb{y}) &=& \operatorname{argmin}_{\mb{x}\in A} ||\mb{y}-\mb{x}||_{\infty}.
       \label{def:projection}
\end{IEEEeqnarray}
Note that this can be computed efficiently (using interior point algorithms) if
$A$ is an affine subspace. Thus given
$\mathbf{\hat{x}}_{\epsilon, \bma{\delta}}$, one can compute
$P_S(\mathbf{\hat{x}}_{\epsilon,\bma{\delta}})$ efficiently. This is simply the
projection of the solution of the r-$(\epsilon,\bma{\delta})$-ISQ decoder on
the solution space $S$ of the ISQ decoder. Before we proceed we observe a certain
property that this projection enjoys.

\begin{lemma}
	There is at least one point of the solution set $S$ of the ISQ within the $\epsilon-l_{\infty}$ ball around $\hat{\mb{x}}_{\epsilon,\bma{\delta}}$, i.e.,
	$||\mb{\hat{x}}_{\epsilon,\bma{\delta}} - P_S (\mb{\hat{x}}_{\epsilon,\bma{\delta}})||_{\infty} \leq \epsilon$.
\label{lemma2}
\end{lemma}

\begin{proof}
We can show this by contradiction. If the claim is false, we would have that 
the function $g(\mb{x}) =||\mb{y} -\mb{H} \mb{x}||^2$ is strictly
convex over the hypercube $\{ \mb{x} : ||\mb{x} -
\mb{x}_{\epsilon,\bma{\delta}}||_{\infty} \leq \epsilon \}$, with  $S$ lying
totally outside the hypercube. By observing that, in such a case, one of the
vertices will have a smaller value for $g(\mb{x})$ than $g(\mb{\hat{x}}_{\epsilon,\bma{\delta}})$, we arrive at a
contradiction.  
\end{proof}

\begin{figure}
  \centering
  \includegraphics{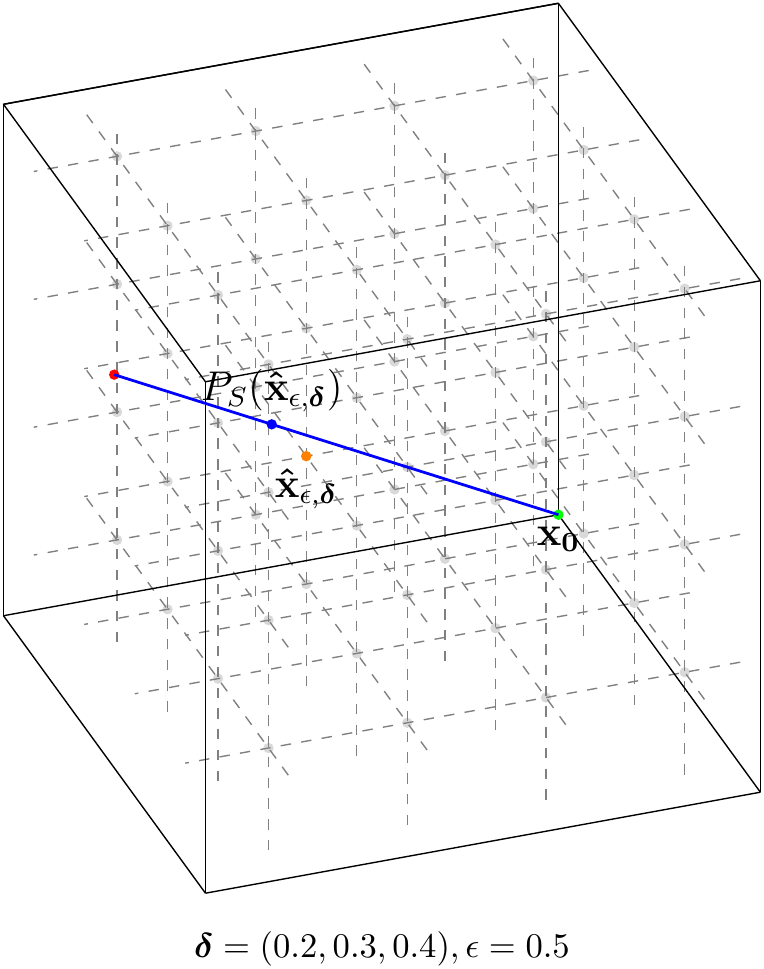}
  \caption{$(\epsilon, \bma{\delta})$-grid for $n=3, \alpha = 2/3$ (dotted grid with gray grid points) for a $\mb{H} \in \mathbb{R}^{2 \times 3}$}
  \label{fig:epsdelta}
\end{figure}
Thus projecting to the solution space $S$ does not change any entry of the vector
$\mathbf{\hat{x}}_{\epsilon,\bma{\delta}}$ by more than $\epsilon$. Since
$\epsilon<0.25$, if $|{\mb{\hat{x}}}_{\epsilon,\bma{\delta},i}| > \epsilon$, the
sign of the corresponding entry of $P_S({\mb{\hat{x}}}_{\epsilon,\bma{\delta}})$
would also be the same as that of $\mathbf{\mb{\hat{x}}}_{\epsilon,\bma{\delta}}$.
The remaining part of the proof is to establish that the sampling of the point
$\mathbf{y} = P_S({\mb{\hat{x}}}_{\epsilon,\bma{\delta}})$ according to
distribution $\hat{f}_S(\mb{y})$ can be done efficiently, i.e. with polynomial complexity (this, in general, is not true for arbitrary multivariate distributions, i.e. \cite{huang2011sampling}). This, together with
the fact that the $|{\mb{\hat{x}}}_{\epsilon,\bma{\delta},i}| $ is less than $
\epsilon$ at most at a sublinear number of coordinates $i$ with overwhelming
probability, establishes the fact that
$P_S({\mb{\hat{x}}}_{\epsilon,\bma{\delta}})$ differs from $\mathbf{x_0}$ in at
most a sublinear number of positions with high probability.

We now relate this randomized projection to the solution of the  r-ISQ decoder.  This simply
takes the affine subspace that is a solution to the ISQ decoder and projects a
random point inside the hypercube on it. Thus  
\begin{IEEEeqnarray}{l+} \text{r-ISQ:}\mb{\hat{x}} =  P_S(\mathbf{x_r}), \mathbf{x_r}\sim \text{Unif}([-1,1]^n) 
	\label{eq:isqr}
\end{IEEEeqnarray}
where $S$ is the solution set of the ISQ decoder. Note that since $S$ is affine
and the sampling is uniform, both can be done efficiently. We now show that the
estimate from this decoder is equal to that of the
r-$(\epsilon,\bma{\delta})$-ISQ decoder for a particular choice of
$f(\bma{\delta})$. This follows from the observation that any distribution on
$\mathbf{x_r}$ would induce a distribution on $S$. This distribution belongs to the family of distributions of the form $\hat{f}_S$ induced by a distribution 
$f(\bma{\delta})$ on $\bma{\delta}$ because the mapping $P_S(\mathbf{\hat{x}}_{\epsilon,\bma{\delta}})$ from $\bma{\delta}$ to $S$ is onto (surjective).
Also, by following the same union bounding technique
used to bound $P_{e,\epsilon,\bma{\delta}}^{k'}$, we get that, for any $k^{''}>0$, the
probability that $\hat{\mb{x}}_{\epsilon,\bma{\delta} }$ has greater than or equal to
$k^{''}n$ entries that are close to zero, (i.e. either $\delta_i-\epsilon,\delta_i,$ or
$\delta_i + \epsilon$) is upper bounded by $\exp(-d_1 n  \log n)$, for some $d_1>0$. Let
\[d_2 = \min(d_1,a).\] Thus we conclude that for a large enough $n$, with
probability at least $1- 2 \exp(-d_2 n \log n)$, the signs of
$P_S(\mb{x_r})$ will match the signs of $\mb{x_0}$ (i.e. the correct
$n$-user codeword) in at least $(1-k^{''}-k^{'}) n $ positions. By choosing
$k^{'}$ and $k^{''}$ small enough we see that the number of mismatches  is
sublinear in the number of transmitting users with overwhelming probability.  The proof of Theorem
\ref{thm:main} is now complete. \qed

To prove Theorem \ref{thm:main2}, we note that, by the symmetry of the system,
the error probability $P_e$ seen by each transmitter is the same. Thus given
any target symbol error rate (SER) $\epsilon_1>0 $, we can choose
$k<\epsilon_1/2$ in Theorem \ref{thm:main}. Then there exists an $n_0$
depending on $k$ such that \[P_e^{k} \leq 2^{-d n \log n} \leq
\frac{\epsilon_1}{2} \,\,\forall n>n_0.\] Then, assuming independent (both
temporally and spatially) channel realizations, we get that the expected number
of errors $E(N_e)$ seen by \emph{all} transmitters in $t$ single shot
transmissions satisfies 

\begin{IEEEeqnarray}{rCl}
	E(N_e) &\leq& n kt + n (1-k) P_e^{k} t\\
  \text{or }	\frac{E(N_e)}{t} &\leq& n \epsilon_1 \,\, \text{for $n$ large enough.}
	\label{eq:finbound}
\end{IEEEeqnarray}
Dividing both sides by $n $  we get that the per-transmitter error probability $P_e$ can be made smaller than $\epsilon_1$ for a large enough $n$.
The proof is now complete.

We now comment on some of the differences from the analysis in \cite{chowdhury2013reliable}. One complication is introduced by the fact that the null space is not empty. Thus the properties of the null space will affect the behaviour of the resulting estimate. In particular, as seen in \cite{bayati2011dynamics}, there does exist a ``bad solution", in the sense that it differs from the correct solution in $\Theta (n)$ symbols. However, we are able to establish that in order to reliably ``clean up" the solution from the ISQ decoder, a random sample would be sufficient. Moreover, it is possible to sample efficiently from this solution set. Thus although computing  $\hat{\mb{x}}_{\epsilon,\bma{\delta}}$
(and thereby $P_S (\hat{\mb{x}}_{\epsilon,\bma{\delta}})$) has exponential complexity, sampling from  $\hat{f}_S$  does not. 
We then propose a simple randomized
solution, and show that this belongs to the family of distributions just mentioned.  
Note
that both the sampling and the projection operation can be done efficiently in
polynomial time.  We show that the distribution on the resulting estimate
$\tilde{f_S}(\bma{\tilde{y}})$ is within the family of distributions
$\{\hat{f}_S(\bma{y}): \bma{y} = P_S(\hat{\mb{x}}_{\epsilon,\bma{\delta}}),
\bma{\delta} \sim f(\bma{\delta}) \}$. This concludes the proof.
\qed

\section{Extensions}
In this section, we point out several extensions to the same ideas that we discussed so far, for more general systems. In particular we focus on the more general kinds of fading distribution, more general constellations, finite blocklength constellations, and faster decay  of the probability of error with $n$. We also indicate how the Theorem \ref{thm:main2} holds not only for a constant $\alpha >0$, but also for an asymptotically vanishing sequence of $\alpha_n$, i.e. \[\alpha_n = \frac{1}{(\log n)^{\xi}} \,\, \textrm{ for }0<\xi <1.\] 
\subsection{General fading distribution}
In the derivation of the proof so far, we assumed i.i.d. $\mathcal{N}(0,1)$ fading for the channel coefficients. We now show how they may be generalized to a much wider class of fading distributions, namely any distribution satisfying the Berry-Esseen bounds on the convergence of the cdf of normalized sums to the gaussian distribution function. 

Before we proceed, we state one version of this lemma.

\begin{lemma}
  Berry-Esseen: Given $N$ i.i.d. random variables $U_1,\ldots,U_N$, with $E[|U_i|^3] \leq \infty$ and $E[|U_i|^2] =\sigma^2$, the following holds for all $x$:
\[\left|P\left(\frac{\sum_{i=1}^N U_i}{\sqrt{N} \sigma} \leq x \right) - \Phi(x) \right|  \leq \frac{K\rho}{\sigma^3 \sqrt{N}}\]
\end{lemma}

In the above $\Phi (x) = 1-Q(x)$ is the cumulative distribution function of a standard normal random variable and $K>0$ is a constant. 

With the above we note that the term appearing in  (\ref{eq:general}) can be expressed as follows:
\begin{IEEEeqnarray}{rCl}
E \left(e^{-t || \sum_{j=1}^n c_j \mb{h_{j}}||^2} \right) &=& E\left(e^{-t n \left|\left|\frac{\sum_{j=1}^n c_j \mb{h_{j}}}{\sqrt{n}}\right|\right|^2}\right)\\
&\overset{(a)}{=}& E\left(e^{-tn\left(\sum_{i=1}^{\alpha n} \left(\sum_{j=1}^{n}c_j H_{i,j}/\sqrt{n}\right)^2\right)}\right)\\
&\overset{(b)}{=}& E\left(e^{-tn \left(\sum_{j=1}^n c_j H_{1,j}/\sqrt{n}\right)^2}\right) ^{\alpha n}.
\end{IEEEeqnarray}
Here $(a)$ follows by decomposing $||\mathbf{y}||^2 = \sum_{i=1}^{\alpha n} y_i^2$, and $(b)$ follows by noting  that $H_{i,j}$ are i.i.d. .
We now observe that an upper bound on the expectation can be written as 
\begin{IEEEeqnarray}{rCl}
E\left(-tn \left(\sum_{j=1}^n c_j H_{1,j}/\sqrt{n}\right)^2\right) &\overset{(c)}{=}& \int_{y=-\infty}^{\infty} e^{-tny^2} p(y) dy \\
&\overset{(d)}{=}& -\int_{u = -\infty}^{\infty} -2unt e^{-tnu^2} P(Y \leq u) du\\
&\overset{(e)}{\leq}& \int_{u=-\infty}^{\infty} 2unt e^{-tnu^2} \Phi(u) du  + \int_{u=-\infty}^{\infty} 2unt e^{-tnu^2} \frac{K \rho}{\sigma^3 \sqrt{n}} du\\  
&\overset{(f)}{\leq}& (1+2nt)^{-1/2} + C /\sqrt{n} \leq (C + t^{-1/2}) n^{-1/2}.
\end{IEEEeqnarray}
In the above $(c)$ follows by defining $p(y)$ to be the density function of $Y = (\sum_{j=1}^n c_j H_{1,j}/\sqrt{n}),$ $(d)$ follows by application of integration by parts, $(e)$ follows from the bound $P(Y \leq u) \leq \Phi(u) + \frac{K \rho}{\sigma^3 \sqrt{n}},$ and $(f)$ follows from actually evaluating the first term and using a trivial bound on the second term i.e. \[ \int_{u=-\infty}^{\infty} 2unt e^{-tnu^2} \frac{K \rho}{\sigma^3 \sqrt{n}} du \leq 2 \int_{u=0}^{\infty}  2unt e^{-tnu^2} \frac{K \rho}{\sigma^3 \sqrt{n}} du.\]
Thus for any fixed $t>0$, we have that
\begin{IEEEeqnarray}{rCl}
(E(-tn (\sum_{j=1}^n c_j H_{1,j}/\sqrt{n})^2)) ^{\alpha n} \leq ((C+t^{-1/2}) n^{-1/2})^{\alpha n} \leq e^{\frac{-\alpha n \log n}{4}} \textrm{ for a large enough $n$.}
\end{IEEEeqnarray}
This, together with the expression in (\ref{eq:general}), choosing $t=1$ gives us the precise asymptotic bound in (\ref{eq:gaussian}). From there on, the remaining claims are the same.
\subsection{General (i.e. non-BPSK) constellations}
For general constellations, the main ideas in the proof remain quite similar, except that the decoder and the proof analysis needs to be slightly different. We first present the generalized decoder and then indicate how the ideas used to establish the result for the BPSK constellation also extend naturally to more general constellations. Let us refer to such a constellation as $\mathcal{M} = \{ m_1, m_2, \ldots, m_N \}$ where $N$ is the number of constellation points.  

We set up some notation first before describing the decoder.

\begin{definition}
The quantizer $\op{Q}$ to a constellation point projects any point $x$ to the nearest constellation point, 
i.e. \[\op{Q}(x) = \operatorname{argmin}_{m_i \in \mathcal{M}} ||x-m_i||_2.\]

\end{definition}

For simplicity  of presentation, $||.||$ refers to the 2-norm unless specified otherwise. 
For a  vector $\mathbf{x}$ of constellation points, $\op{Q}(\mathbf{x})$ projects each coordinate of $\mathbf{x}$ to the nearest constellation point in $\mathcal{M}$, i.e. \[(\op{Q}(\mathbf{x}))_i = \op{Q}(x_i).\]

Given this notation, the ISQ decoder defined earlier, is equivalent to 
\begin{IEEEeqnarray}{rCl}
\mb{\hat{x}}
	&=&\operatorname{Q}(\operatorname{argmin}_{||\mb{x}||_{\infty} \leq \max_{m_i \in \mathcal{M}} ||m_i||}
	||\mb{y}-\mb{H x}|| ^2).
\label{eq:genISQ}
\end{IEEEeqnarray}

This reduces to (\ref{eq:3}) when $\mathcal{M} =\{-1,+1\}$. The definition of the randomized ISQ decoder follows along very similar lines. We pick a point $\mb{x_r}$ randomly from an uniform distribution over the set 
\[B_n = \{ \mb{x}: ||\mb{x}||_{\infty} \leq \max_{m_i \in \mathcal{M}} ||m_i|| \}.\] Then we project $\mb{x}$ on the (possibly non singleton) solution set $S = \operatorname{argmin}_{\mb{x}:||\mb{x}||_{\infty} \leq \max_{m_i \in \mathcal{M}} ||m_i||}
	||\mb{y}-\mb{H x}|| ^2$ of the ISQ decoder, i.e. we return
\[\hat{\mb{x}} = \op{Q}(P_S(\mb{x_r})), \mb{x_r} \sim \operatorname{Unif}(B_n).\]
In the above, the projection operation is the same as that introduced in (\ref{def:projection}). We can show that with the above decoders, the same conclusions  that we derived in Theorem \ref{thm:main2} continue to hold. The proof however needs some generalization of some of the ingredients involved in the proof. We point out such generalizations in the following. 

A critical step towards obtaining Theorem \ref{thm:main2} was the use of the $r-(\epsilon, \bma{\delta})$-grid detector. For a general constellation $\mathcal{M}$, define the scalar grid  as  (using $B_n$ as we defined it earlier) \[\mathcal{G}_{1,\epsilon} = \{ \{g_1, g_2, \ldots, g_{N_{\epsilon}}\}: \textrm{ For any $x \in B_1$}, ||x- g_i||_{\infty} \leq \epsilon \textrm{ for some } i; \\ ||g_i-g_j||_{\infty} > \epsilon\, \forall i \neq j; ||g_i||_2  \leq \max_{m_j \in \mathcal{M}} ||m_j||_2 \, \forall i\}. \]
Note that $\mathcal{G}_{1,\epsilon}$ is not unique. We illustrate possible scalar grids for the BPSK constellation considered earlier and a 2D constellation (e.g. 4-PSK) in Fig. \ref{fig:grid}.\\
\begin{figure}
\begin{subfigure}{0.450\linewidth}
\centering
  \begin{tikzpicture}[scale=0.8]
    \fill[blue] (0,0) circle [radius=2pt, minimum size=0.01cm, blue];
   \fill[blue] (6,0) circle [radius=2pt, minimum size=0.01cm, blue];
   \foreach \i in {1,2,3,4,5}
   {
\fill[red] (\i,0) circle [radius=2pt, minimum size=0.01cm];
   }
\draw[<->]  (-1,0) --  (7,0);
\node at (0,-0.5) {-1};
\node at (6,-0.5) {1};
\draw[-] (3,-4) -- (3,-4.0001);
\draw[-] (3,4) -- (3,4.0001);
  \end{tikzpicture}
\caption{$\mathcal{G}_{1,\epsilon}$  (red dots) for $\epsilon=1/3$ for BPSK (blue dots)}
\end{subfigure}
\begin{subfigure}{0.450\linewidth}
\centering

  \begin{tikzpicture}[scale=0.8]
    \fill[blue] (0,0) circle [radius=2pt, minimum size=0.01cm, blue];
   \fill[blue] (6,0) circle [radius=2pt, minimum size=0.01cm, blue];
   \fill[blue] (3,3) circle [radius=2pt, minimum size=0.01cm, blue];
   \fill[blue] (3,-3) circle [radius=2pt, minimum size=0.01cm, blue];
   \foreach \i in {1,2,3,4,5}
   {
\fill[red] (\i,0) circle [radius=2pt, minimum size=0.01cm];
   }
 \foreach \i in {1,2,3,4,5}
   {
\fill[red] (\i,1) circle [radius=2pt, minimum size=0.01cm];
\fill[red] (\i,-1) circle [radius=2pt, minimum size=0.01cm];
\fill[red] (\i,2) circle [radius=2pt, minimum size=0.01cm];
\fill[red] (\i,-2) circle [radius=2pt, minimum size=0.01cm];
   }

\draw[<->]  (-1,0) --  (7,0);
\draw[<->]  (3,4) -- (3,-4);
\node at (0,-0.5) {-1};
\node at (6,-0.5) {1};
  \draw[blue!50!white!50] (3,0) circle(3cm);

  \end{tikzpicture}
\caption{$\mathcal{G}_{1,\epsilon}$(red dots) for $\epsilon = 1/3$ for  $4$-PSK (blue dots)}

\end{subfigure}
\caption{Possible scalar grids $\mathcal{G}_{1,\epsilon}$ for different constellations}
\label{fig:grid} 
\end{figure}
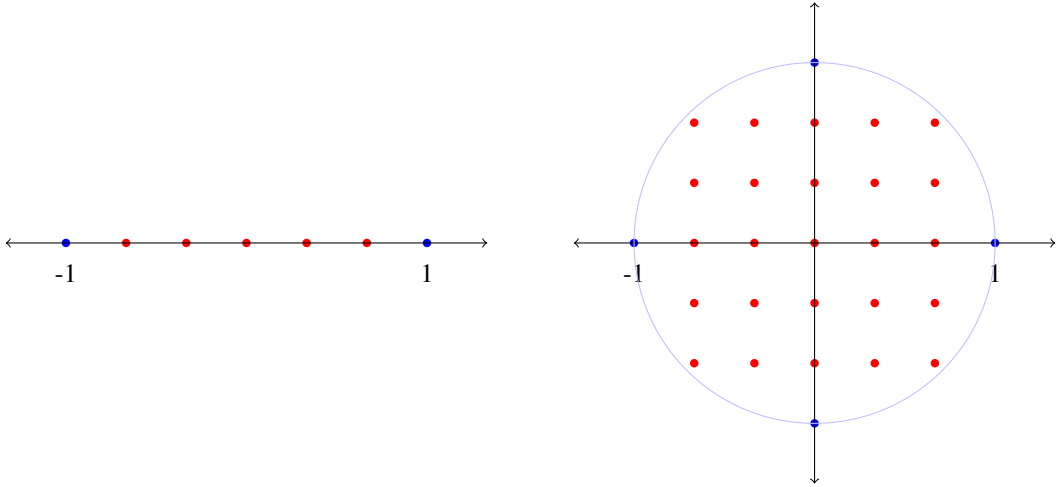 We define the perturbed grid $\mathcal{G}_{1,\epsilon,\delta} = \{{g_1+\delta ,g_2 + \delta,\dots,g_{N_{\epsilon}}+\delta}: g_i \in \mathcal{G}_{1,\epsilon}\,\, \forall i \}$. The perturbed grid in $n$ dimensions is then simply defined as 
\[\mathcal{G}_{n,\epsilon, \bma{\delta}} = \{ \mb{x}: x_i \in \mathcal{G}_{1,\epsilon,\delta_i} \,\, \forall i \}.\]

The $(\epsilon, \bma{\delta})$-ISQ decoder would then be 
\begin{IEEEeqnarray}{l+} \text{$(\epsilon,\bma{\delta})$-ISQ:} \,\mb{\hat{x}}_{\epsilon,\bma{\delta}}
	=\operatorname{Q}(\operatorname{argmin}_{\mb{x} \in \mathcal{G}_{n,\epsilon,\bma{\delta}} }
	||\mb{y}-\mb{H x}|| ^2). \nonumber
	\label{eq:3ageneral}
\end{IEEEeqnarray}

A bound on the error event for the general constellation $\mathcal{C}$ can be had in terms of the minimum distance of the constellation $\mathcal{C}$  defined as $d_{min} = \min_{c_i,c_j \in \mathcal{C}, i \neq j} ||c_i - c_j||_2$. A first step towards that is the observation that Lemma \ref{lemma1}  holds by observing that for $k^{'}n$ errors, \[\sum_j ||c_j||^2 \geq k^{'}n d_{min}^2.\] Similarly Lemma \ref{lemma2} will follow from the observation that $g (\mb{x}) = ||\mb{y} - \mb{H x}||^2$ is convex (independent of the constellation used).

Combining these two observations, we have the result that the estimate from the r-$(\epsilon, \bma{\delta})$ ISQ decoder will have greater than or equal to $k^{'}n$ errors with probability less than $\exp(-a n \log n)$ for some $a>0$. This, together with the fact that for a finite dimensional constellation, $||c_i||_{\infty} \leq \epsilon \Rightarrow ||c_i||_2 \leq \tilde{K} \epsilon $ for some $\tilde{K}>0$ (and vice versa), gives us the result that \[||Q(P_S(\mb{x_r}))-\mb{x_0}||_0 \leq kn,\] with probability at least $1- \exp(-d n \log n)$ with $d>0$,  for any $k>0$, and a large enough $n$. 

This establishes Theorem \ref{thm:main} from which Theorem \ref{thm:main2} follows by the symmetry in the system model.   




\subsection{Finite blocklength constellation design}
A finite blocklength constellation over $T$ time slots (together with a block fading model i.e. a model with the channel matrix $\mb{H}$ remaining constant over $T$ time slots) can be thought of as a general constellation within a single shot transmission model. Thus by the generalization of Theorem \ref{thm:main2} shown in the previous section for arbitrary constellations, we get that the same results hold for arbitrary finite blocklength constellations too. 
\subsection{Provably faster  decay of the error probability}
In the proofs so far, we demonstrated that the number of symbol errors in the \emph{n-user} codeword is less than $kn$ for any $k>0$, i.e. with high probability, the number of errors is eventually sublinear. By using the same techniques to derive the above result, we can also establish the result that the number of errors is less than $n^{\epsilon}$ for any $0<\epsilon <1.$ This would not change any of the conclusions of the previously stated theorems.

\subsection{Asymptotically vanishing sequence of $\alpha_n$}
We restate a version of Theorem \ref{thm:main2} for this case.

\begin{theorem}
The probability of error with the r-ISQ decoder seen by any
transmitting user vanishes in the limit of an asymptotically large number of
transmitters, with the per-transmitter number of receive antennas $\alpha_n$ scaling with the number $n$ of transmitting users like  \[\alpha_n = \frac{1}{(\log n)^{\xi}},\,\, 0<\xi <1.\]
\end{theorem}

The critical step towards proving this theorem is the observation that an analogue of Lemma \ref{lemma1} continues to hold in this case with  the following modifications
\begin{lemma} For any $i> k'n$, there exists an $n_0$ and an $a>0$, such that for
	all $n>n_0$, \[P(|| \sum_{j=1}^n c_j \mb{h_{j}}||^2 <  n (\log n)^{1-\xi}) \leq
		\exp(- n (\log n)^{1-\xi}).\] 
                \label{lemma1mod}
\end{lemma}
The proof follows along lines very similar to the ones used to derive the original lemma \ref{lemma1}. Based on this observation Theorem \ref{thm:main} holds with the following modification
\begin{theorem}
	Under r-ISQ decoding, for $m = \alpha_n n$ defined earlier,  and any constant $k>0$, there exists a $d>0$ such that
	for all sufficiently large  $n$,
	\[P_{e}^{k} \leq 2^{-d n (\log n)^{1-\xi}}.\]
	\label{thm:mainseq}
\vspace{-0.5cm}
\end{theorem}
 Theorem \ref{thm:main2} would also hold unmodified in this case.
\section{Conclusions and Future Work} 

We have considered an uplink communication system in a rich scattering environment
with a large number of non-cooperating transmitters and a large number of antennas at the receiver. The transmitters  send bits to the receiver \emph{without coding}. The receiver does joint decoding of the
noisy received signal from all users  using a  relaxation of the maximum likelihood (ML)
decoder. We call this  technique the interval search and quantize (ISQ) decoder. Since the solution may not be unique in general, we have proposed an efficient randomization scheme (i.e. the r-ISQ decoder) which will still allow us to have reliable estimates from the solution set.  Under general assumptions about the fading distribution of the channel coefficients, we have shown
that with the r-ISQ decoder, for a large enough system size, the error probability that each user sees
is vanishingly small even with the per-transmitter number of receiver antennas being arbitrarily small.  


In
spite of these promising asymptotic properties of the efficient box-constrained
decoders, we pay a price in the finite $n$ behaviour
with respect to the rate of decay of the error probability. The decay rates achievable are at best
polynomial.  Thus, the question of how large the system size needs to be for
the diversity-induced reliability to kick in remains a pertinent question and needs further investigation. Also, the issue of
how practical constraints such as limited diversity in large systems or imperfect
channel knowledge would affect these results remains a topic to be investigated.

\section*{Acknowledgements}

This work was supported by a 3Com Corporation Stanford Graduate Fellowship,
by the NSF Center for Science of Information (CSoI): NSF-CCF-0939370, and by a
gift from Cablelabs. The authors acknowledge helpful discussions and insights from
Tsachy Weissman, in particular with respect to the proof of Theorem  \ref{thm:main}. The first
author would also like to acknowledge helpful discussions with Yash Deshpande,
Stefano Rini, Alexandros Manolakos and Nima Soltani. 

\bibliography{reference_jabref}

\begin{thebibliography}{10}
\providecommand{\url}[1]{#1}
\csname url@samestyle\endcsname
\providecommand{\newblock}{\relax}
\providecommand{\bibinfo}[2]{#2}
\providecommand{\BIBentrySTDinterwordspacing}{\spaceskip=0pt\relax}
\providecommand{\BIBentryALTinterwordstretchfactor}{4}
\providecommand{\BIBentryALTinterwordspacing}{\spaceskip=\fontdimen2\font plus
\BIBentryALTinterwordstretchfactor\fontdimen3\font minus
  \fontdimen4\font\relax}
\providecommand{\BIBforeignlanguage}[2]{{%
\expandafter\ifx\csname l@#1\endcsname\relax
\typeout{** WARNING: IEEEtran.bst: No hyphenation pattern has been}%
\typeout{** loaded for the language `#1'. Using the pattern for}%
\typeout{** the default language instead.}%
\else
\language=\csname l@#1\endcsname
\fi
#2}}
\providecommand{\BIBdecl}{\relax}
\BIBdecl

\bibitem{cover1991elements}
T.~Cover and J.~Thomas, \emph{{E}lements of {I}nformation {T}heory}.\hskip 1em
  plus 0.5em minus 0.4em\relax Wiley Online Library, 1991, vol.~6.

\bibitem{caire2004suboptimality}
G.~Caire, D.~Tuninetti, and S.~Verd{\'u}, ``{S}uboptimality of {TDMA} in the
  {L}ow-{P}ower {R}egime,'' \emph{IEEE Transactions on Information Theory},
  vol.~50, no.~4, pp. 608--620, 2004.

\bibitem{tse2005fundamentals}
D.~Tse and P.~Viswanath, \emph{{F}undamentals of {W}ireless
  {C}ommunication}.\hskip 1em plus 0.5em minus 0.4em\relax Cambridge University
  Press, 2005.

\bibitem{chowdhury2014uncodedfundamental}
M.~Chowdhury and A.~Goldsmith, ``{R}eliable {U}ncoded {C}ommunication in the
  {SIMO} {MAC},'' 2014, submitted to IEEE Transactions on Information Theory.

\bibitem{chowdhury2013reliable}
M.~Chowdhury, A.~Goldsmith, and T.~Weissman, ``{R}eliable {U}ncoded
  {C}ommunication in the {SIMO MAC} {V}ia {L}ow-{C}omplexity {D}ecoding,'' in
  \emph{Proceedings of ISIT}.\hskip 1em plus 0.5em minus 0.4em\relax IEEE,
  2013.

\bibitem{mainakch2012uncoded}
------, ``{T}he {Per-U}ser {N}umber of {R}eceive {A}ntennas in {U}ncoded
  {N}on-{C}ooperating {T}ransmissions {C}an {B}e {A}rbitrarily {S}mall,'' in
  \emph{Proceedings of the 50th Annual Allerton Conference on
  Communication,Control, and Computing, Monticello, IL}, 2012.

\bibitem{boyd2004convex}
S.~Boyd and L.~Vandenberghe, \emph{{C}onvex {O}ptimization}.\hskip 1em plus
  0.5em minus 0.4em\relax Cambridge University Press, 2004.

\bibitem{donoho2010counting}
D.~L. Donoho and J.~Tanner, ``{C}ounting the {F}aces of {R}andomly-{P}rojected
  {H}ypercubes and {O}rthants, with {A}pplications,'' \emph{Discrete \&
  Computational Geometry}, vol.~43, no.~3, pp. 522--541, 2010.

\bibitem{bayati2011dynamics}
M.~Bayati and A.~Montanari, ``{T}he {D}ynamics of {M}essage {P}assing on
  {D}ense {G}raphs, with {A}pplications to {C}ompressed {S}ensing,'' \emph{IEEE
  Transactions on Information Theory}, vol.~57, no.~2, pp. 764--785, 2011.

\bibitem{huang2011sampling}
Z.~Huang and S.~Kannan, ``{O}n {S}ampling from {M}ultivariate
  {D}istributions,'' in \emph{Approximation, Randomization, and Combinatorial
  Optimization. Algorithms and Techniques}.\hskip 1em plus 0.5em minus
  0.4em\relax Springer, 2011, pp. 616--627.

\end{thebibliography}
\bibliographystyle{IEEEtran}

\end{document}